\documentclass{article}

\usepackage{PRIMEarxiv}

\usepackage[utf8]{inputenc} 
\usepackage[T1]{fontenc}    
\usepackage{hyperref}       
\usepackage{url}            
\usepackage{booktabs}       
\usepackage{amsfonts}       
\usepackage{nicefrac}       
\usepackage{microtype}      
\usepackage{lipsum}
\usepackage{fancyhdr}       
\usepackage{graphicx}       
\usepackage[numbers]{natbib}
\usepackage[table]{xcolor}
\usepackage{algorithm}
\usepackage[noend]{algorithmic}
\usepackage{amsthm,amssymb,amsmath}
\usepackage{multirow}
\usepackage{cleveref} 
\graphicspath{{media/}}     

\newtheorem{theorem}{Theorem}
\newtheorem{proposition}{Proposition}

\title{Rethinking Security of Diffusion-based Generative Steganography}

\author{
  Jiahao Zhu, Zixuan Chen \\
  School of Computer Science and Engineering\\
  Sun Yat-sen University \\
   \And
  Jiali Liu \\
  Zhongshan School of Medicine \\
  Sun Yat-sen University\\
  \And
  Lingxiao Yang \\
  School of Systems Science and Engineering \\
  Sun Yat-sen University\\
  \And
  Yi Zhou\\
  Zhongshan School of Medicine \\
  Sun Yat-sen University\\
\And
  Weiqi Luo, Xiaohua Xie\\
  School of Computer Science and Engineering \\
  Sun Yat-sen University\\
}

\begin{document}
\maketitle

\begin{abstract}
  Generative image steganography is a technique that conceals secret messages within generated images, without relying on pre-existing cover images. Recently, a number of diffusion model-based generative image steganography (DM-GIS) methods have been introduced, which effectively combat traditional steganalysis techniques. In this paper, we identify the key factors that influence DM-GIS security and revisit the security of existing methods. Specifically, we first provide an overview of the general pipelines of current DM-GIS methods, finding that the noise space of diffusion models serves as the primary embedding domain. Further, we analyze the relationship between DM-GIS security and noise distribution of diffusion models, theoretically demonstrating that any steganographic operation that disrupts the noise distribution compromise DM-GIS security. Building on this insight, we propose a Noise Space-based Diffusion Steganalyzer (NS-DSer)—a simple yet effective steganalysis framework allowing for detecting DM-GIS generated images in the diffusion model noise space. We reevaluate the security of existing DM-GIS methods using NS-DSer across increasingly challenging detection scenarios. Experimental results validate our theoretical analysis of DM-GIS security and show the effectiveness of NS-DSer across diverse detection scenarios.
\end{abstract}

\section{Introduction}
Image steganography conceals secret messages within innocuous images for covert communication \cite{luo2024comprehensive, MANDAL20221451}. Traditional steganographic methods embed secret bits in the spatial \cite{li2015strategy,sedighi2015content,zhang2016decomposing,huang2023automatic,huang2023steganography} or frequency domains \cite{huang2012new,li2019jpeg,wang2022jpeg} of natural images (\emph{a.k.a.} cover images). However, this process inevitably alters these cover images, rendering them susceptible to detection. Moreover, these algorithms have limited embedding capacity. To overcome these limitations, a novel paradigm known as generative steganography has emerged. It leverages generative models, such as GANs \cite{liu2017coverless,wang2018sstegan,wei2022generativegane} and flow models \cite{wei2022generativeflow,zhou2022secret}, to generate stego images directly from secret messages, eliminating the need for pre-existing cover images. 

Recently, a growing number of generative image steganography methods based on diffusion models (DM-GIS) have been proposed \cite{diffusion-stego,establishing,improved,gaussian-shading,pulsar,stegaddpm}. Compared to GAN-based and flow-based counterparts, DM-GIS methods offer the advantage of being training-free. Furthermore, owing to the powerful generative capabilities of diffusion models \cite{ddpm,ldm}, the stego images they generate exhibit higher visual plausibility. With the rapid development of AIGC, pre-trained diffusion models have become increasingly accessible online, significantly facilitating the implementation of DM-GIS. Consequently, there is an increasingly urgent demand for effective detection methods. 

However, recent studies \cite{stegaddpm,diffusion-stego,pulsar,establishing,improved} have shown that current DM-GIS methods can effectively evade universal image steganalysis, achieving near-perfect empirical security. This phenomenon can be attributed to two key factors:
\textbf{(1) Unique embedding space}. Unlike traditional image steganography, which embeds secret messages directly in the image domain, DM-GIS utilizes the unique structural properties of diffusion models to embed messages in their noise space. As a result, conventional steganalysis methods that focus on the image domain \cite{xunet,yenet,srnet,wang2022jpeg} are significantly less effective against DM-GIS techniques.
\textbf{(2) Heterogeneous data sources}. In DM-GIS scenarios, cover images encompass not only natural images but also those generated by diffusion models, which belong to different data distributions in practice. Moreover, both cover and stego images can be produced using various diffusion models, sampling steps, samplers, and guidance scales, further complicating the task of distinguishing between cover and stego images.

In this paper, we rethink the steganographic security of DM-GIS and reevaluate current methods. We begin by reviewing existing DM-GIS pipelines, observing that secret messages are embedded within either the initial or intermediate noise vectors of diffusion models. This observation motivates us to link the noise space of diffusion models to DM-GIS security. Specifically, we provide a theoretical proof that any steganographic operation that disrupts the noise distribution of diffusion models will compromise DM-GIS security. We also examine the relationship between DM-GIS security and message extraction accuracy, demonstrating that, without modifying message embedding schemes, DM-GIS methods with high extraction accuracy are more susceptible to detection. Building on these insights, we introduce the \textbf{N}oise \textbf{S}pace-based \textbf{D}iffusion \textbf{S}teganalyser (NS-DSer), a straightforward yet powerful steganalysis framework for re-evaluating existing DM-GIS methods. It performs detection within the noise space of diffusion models and does not need to know the specific diffusion model configurations used by steganographers. This framework offers two merits: 
\textbf{(1)} It simplifies the complex task of distinguishing between cover and stego images by transforming it into a more tractable problem of differentiating statistical distributions, making our method both straightforward and effective.
\textbf{(2)} It enhances the practicality of steganalysis by making it insensitive to heterogeneous data sources. NS-DSer consists of two stages: a deterministic condition-free diffusion process and statistical feature extraction. In the first stage, we adopt deterministic sampling to recover the initial noise vectors of the given images without condition guidance. In the second stage, we extract statistical features from both the noise space and its corresponding transformed domain for classification.

The contributions of this paper are summarized as follows:
\begin{itemize}
    \item We outline the general pipelines of existing DM-GIS methods and provide a theoretical analysis that reveals the key factors influencing steganographic security.
    \item We propose NS-DSer, a simple yet effective DM-GIS steganalysis framework for re-evaluating existing DM-GIS methods. It efficiently addresses the challenges of unique embedding spaces and heterogeneous data sources.
    \item We design four progressively challenging scenarios to evaluate current DM-GIS methods. Experiments demonstrate that NS-DSer outperforms traditional steganalyzers in all four scenarios.
\end{itemize}
This paper is organized as follows. \Cref{sec:related work} briefly reviews existing DM-GIS methods and traditional steganalysis. \Cref{sec:preliminary} introduces diffusion models and the definition of steganographic security. \Cref{sec:theoretical_analysis_security} provides a theoretical analysis of DM-GIS security. \Cref{sec:nsdser} introduces our proposed steganalyzer NS-DSer. \Cref{sec:exp_results} reports the experimental results. Finally, \Cref{sec:conclusion} concludes the paper and outlines potential future directions.
\section{Related Work}
\label{sec:related work}
\subsection{Diffusion Model-based Generative Steganography}
Diffusion models are widely used in generative steganography. Peng \emph{et al.} \cite{stegaddpm} introduced StegaDDPM, which embeds secret bits using the probability distribution between intermediate states and generated images. Jois \emph{et al.} \cite{pulsar} proposed Pulsar, concealing messages with variance noise in the denoising process. Peng \emph{et al.} \cite{ldstega} enhanced StegaDDPM by adding a truncated Gaussian encoding mechanism and integrating Latent Diffusion Models (LDMs) \cite{ldm} as the backbone. Kim \emph{et al.} \cite{diffusion-stego} introduced a two-stage framework with four bit projection methods—Message to Noise (MN), Message to Binary (MB), Message to Centered Binary (MC), and Multi-bits—balancing extraction accuracy, security, and image quality. Yang \emph{et al.} \cite{gaussian-shading} developed Gaussian shading for copyright protection in LDMs. Hu \emph{et al.} \cite{establishing} improved robustness by embedding bits in LDMs' initial noise transform domain, while Zhou \emph{et al.} \cite{improved} used the discrete cosine transform domain. Yu \emph{et al.} \cite{cross} introduced CRoSS for image hiding with a controllable, reversible key system, and Yang \emph{et al.} \cite{diffstega} addressed key leakage in CRoSS with password-dependent key management.
\subsection{Image Steganalysis}
Traditional steganalysis relies on handcrafted features like SRM \cite{srm} and GFR \cite{gfr}, but these methods struggle with advanced techniques. Deep learning, particularly CNNs, has become popular for superior feature extraction. Qian \emph{et al.} \cite{gncnn} introduced the GNCNN model, matching classical techniques, while XuNet \cite{xunet} surpassed SRM in accuracy. Ye \emph{et al.} \cite{yenet} improved steganalysis with YeNet, outperforming SRM and maxSRM \cite{max-srm}. Boroumand \emph{et al.} \cite{srnet} proposed SRNet, a deep residual network for better detection of spatial and JPEG steganography. Deng \emph{et al.} \cite{deng} simplified this with a four-group CNN. Zhang \emph{et al.} \cite{zhunet} introduced ZhuNet, using depth-wise separable convolutions. Recent studies have focused on color image steganalysis, with Zeng \emph{et al.} \cite{wisernet} proposing WISERNet for better accuracy. Butora \emph{et al.} \cite{butora2021pretrain} evaluated ImageNet-pretrained CNNs, finding EfficientNet \cite{efficientnet} effective for JPEG but not spatial domains. You \emph{et al.} \cite{you2020siamese} proposed a Siamese CNN architecture, while Wei \emph{et al.} \cite{ucnet} introduced UCNet for advanced spatial steganalysis. Most recently, Wei \emph{et al.} proposed PENet \cite{penet}, a model that detects stego images of varying sizes without retraining.
\begin{figure}
    \centering
    \includegraphics[width=.9\textwidth]{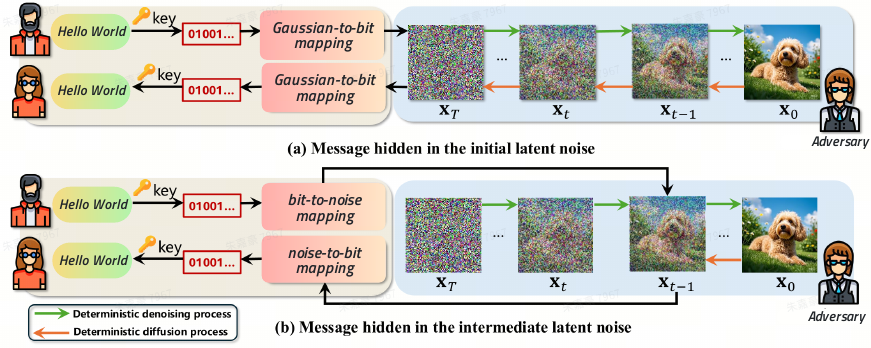}
     \caption{Two types of DM-GIS frameworks: (a) The secret messages are embedded in the initial noise $\mathbf{x}_T$, which is typically a Gaussian white noise; and (b) The secret messages are concealed within the intermediate noise $\mathbf{x}_t$ at a certain timestep $t$ in the deterministic denoising process.}
     \label{fig:framework}
\end{figure}
\section{Preliminary}
\label{sec:preliminary}
\subsection{Diffusion Models}
Diffusion models \cite{scored-generative,ddim,karras} consist of a diffusion and a denoising process. The diffusion process is modeled by an It$\hat{\text{o}}$ SDE with transition kernel $\mathbb{P}(\mathbf{x}_t|\mathbf{x}_0) := \mathcal{N}(\mathbf{x}_t; \alpha_t \mathbf{x}_0, \sigma_t^2 \mathbf{I})$. The denoising process is described by the reversed SDE:
\begin{equation}
    \label{eq:sde_reverse}
    \mathrm{d}\mathbf{x} = \left[ f(\mathbf{x}, t) - g(t)^2 \nabla_\mathbf{x} \log \mathbb{P}_t(\mathbf{x}) \right] \mathrm{d}t + g(t) \mathrm{d}\boldsymbol{\omega},
\end{equation}
where $f(\mathbf{x}, t)$ and $g(t)$ are the drift and diffusion coefficients, and $\boldsymbol{\omega}$ is Brownian motion. A neural network $\boldsymbol{\epsilon}_{\phi}(\mathbf{x}, t)$ approximates the scaled score function $-\sigma_t \nabla_{\mathbf{x}} \log \mathbb{P}_t(\mathbf{x})$. 

An image $\mathbf{x}_0$ can be obtained by solving \Cref{eq:sde_reverse} starting from noise $\mathbf{x}_T \sim \mathcal{N}(\mathbf{0}, \sigma_T^2 \mathbf{I})$. Alternatively, image generation can be done via a PF-ODE derived from \Cref{eq:sde_reverse}:
\begin{equation}
    \label{eq:pf-ode}
    \mathrm{d}\mathbf{x} = \left[ f(\mathbf{x}, t) - \frac{1}{2} g(t)^2 \nabla_{\mathbf{x}} \log \mathbb{P}_t(\mathbf{x}) \right] \mathrm{d}t.
\end{equation}
\subsection{Diffusion Model-based Stegosystem}
\label{sec:diffusion_stegosystem}
In stegosystems, steganographers embed secret messages in carriers like images, which are transmitted over a public channel, while a passive adversary, $\mathcal{A}$, seeks to detect hidden communication. The goal is for steganographers to evade $\mathcal{A}$ and transmit encoded messages. In Diffusion Model-based Generative Steganography (DM-GIS), steganographers encrypt messages using keys and employ reversible bit-to-noise mappings for hiding the message. Leveraging the invertibility of diffusion models \cite{ddim,ddpm,ldm}, steganographers establish an invertible mapping between stego images and secret messages. Some methods \cite{pulsar,stegaddpm,ldstega} use the SDE with shared random seeds for invertibility, while others \cite{diffusion-stego,gaussian-shading,establishing,improved,cross,diffstega} use the PF-ODE. \Cref{fig:framework} illustrates the two DM-GIS frameworks.
\subsection{Definition of Steganographic Security}
\label{sec:definition_security}
Steganographic security can be defined in two ways. Cachin \cite{cachin1998information} defined security based on information theory, quantifying it with the Kullback-Leibler (KL) divergence between the cover distribution $\mathbb{P}_c$ and the stego distribution $\mathbb{P}_s$. A system is $\epsilon$-secure if 
\begin{equation}
    \label{eq:security1}
    D_{KL}(\mathbb{P}_c \parallel \mathbb{P}_s) \leq \epsilon.
\end{equation}
If $\epsilon = 0$, the system is perfectly secure. Hopper \emph{et al.} \cite{hopper2002provably} defined security from a complexity theory perspective. A stegosystem is \((N_q, T_r, \epsilon)\)-secure against a hidden-text attack if
\begin{equation}
    \label{eq:security2}
    \max_{N_q, T_r} \left| \mathbb{P} \left( \mathcal{A}^{O_\mathcal{S}(\mathbf{k}, \mathbf{m})} = 1 \right) - \mathbb{P} \left(\mathcal{A}^{O_\mathcal{C}} = 1 \right) \right| < \epsilon.
\end{equation}
$O_\mathcal{C}$ and $O_\mathcal{S}$ are oracles for the cover and stego distributions. $\mathbf{k}$ is a randomly chosen key and $\mathbf{m}$ is a secret message. The adversary $\mathcal{A}$ is allowed to make at most $N_q$ queries within a runtime of $T_r$.
\section{Theoretical Analysis of DM-GIS Security}
\label{sec:theoretical_analysis_security}
In \Cref{subsec:theorical_analysis}, we provide a steganographic security analysis of Diffusion Model-based Generative Steganography (DM-GIS), uncovering the relationship between DM-GIS security, noise distribution of diffusion models, and message extraction accuracy.
\subsection{Connecting Diffusion Model Noise Distribution to DM-GIS Security}\label{subsec:theorical_analysis}
Previous diffusion models \cite{ddim,scored-generative} have stated that one can theoretically construct an invertible mapping between the latent $\mathbf{x}_T$ and the image $\mathbf{x}_0$ driven by \Cref{eq:ode_solver}. For diffusion models, the generation process can be defined by $\mathbf{x}_0=F(\mathbf{x}_T)$, which can be decomposed into $N_s$ piecewise integrals \cite{dpm,ddim},
\begin{equation}
    \label{eq:pf-ode-de}
    \mathbf{x}_t = \mathbf{x}_{t+1}+\int_{t+1}^{t}\left[ f(\mathbf{x}, t) - \frac{1}{2} g(t)^2 \nabla_{\mathbf{x}} \log \mathbb{P}_t(\mathbf{x}) \right] \mathrm{d}t.
\end{equation}
Let $F^{-1}$ denote the corresponding reverse process; we then present the following theorem:
\begin{theorem}\label{theorem:t0}
Let $\mathbb{Q}_c$ and $\mathbb{Q}_s$ denote the distributions of normal diffusion model noise and noise containing secret messages, respectively. We denote $\mathbb{P}_s$ and $\mathbb{P}_c$ as the stego image distribution and normally generated image distribution, respectively. Theoretically, the following equality holds: $D_{KL}(\mathbb{P}_c\| \mathbb{P}_s) = D_{KL}(\mathbb{Q}_c \| \mathbb{Q}_s).
$
\end{theorem}
\begin{proof}
    Assuming $\mathcal{X}$ denotes the whole image space, we have
\begin{equation}
            D_{KL}(\mathbb{P}_c(\mathbf{x})||\mathbb{P}_s(\mathbf{x}))=\sum_{\mathbf{x}\in\mathcal{X}}\mathbb{P}_c(\mathbf{x})\log\frac{\mathbb{P}_c(\mathbf{x})}{\mathbb{P}_s(\mathbf{x})}
            =\sum_{\mathbf{g}\in{F^{-1}(\mathcal{X})}}\sum_{\mathbf{x}\in F(\mathbf{g})}\mathbb{P}_{c}(\mathbf{x})\log\frac{\mathbb{P}_c(\mathbf{x})}{\mathbb{P}_s(\mathbf{x})}.
\end{equation}
Let $F(\mathbf{g})=\{\mathbf{x}_1,\cdots,\mathbf{x}_k\}$, we have
\begin{equation}
    \begin{aligned}
                &\sum_{\mathbf{x}\in F(\mathbf{g})}\mathbb{P}_c(\mathbf{x})\log\frac{\mathbb{P}_c(\mathbf{x})}{\mathbb{P}_s(\mathbf{x})}=\sum_{i=1}^{k}\mathbb{P}_c(\mathbf{x}_i)\log\frac{\mathbb{P}_c(\mathbf{x}_i)}{\mathbb{P}_s(\mathbf{x}_i)}.
    \end{aligned} 
\end{equation}
According to the log sum inequality, we can easily get
\begin{equation}
    \sum_{i=1}\mathbb{P}_c(\mathbf{x}_i)\log\frac{\mathbb{P}_c(\mathbf{x}_i)}{\mathbb{P}_s(\mathbf{x}_i)}\geq\sum_{i=1}^{k}\mathbb{P}_c(\mathbf{x}_i)\log\frac{\sum_{j=1}^{k}\mathbb{P}_c(\mathbf{x}_j)}{\sum_{j=1}^{k}\mathbb{P}_s(\mathbf{x}_j)}.
\end{equation}
As $\mathbf{g}\in F^{-1}(\mathcal{X})$, $\mathbb{Q}_c(\mathbf{g})=\sum_{\mathbf{x}\in F^{-1}(\mathbf{g})}\mathbb{P}_c(\mathbf{x})$, and similarly, $\mathbb{Q}_s(\mathbf{g})=\sum_{\mathbf{x}\in F^{-1}(\mathbf{g})}\mathbb{P}_s(\mathbf{x})$. Accordingly, we have
\begin{equation}
        D_{KL}(\mathbb{P}_c||\mathbb{P}_s)\geq \\
        \sum_{\mathbf{g}\in{F^{-1}(\mathcal{X})}}\sum_{i=1}^{k}\mathbb{P}_c(\mathbf{x}_i)\log\frac{\sum_{j=1}^{k}\mathbb{P}_c(\mathbf{x}_j)}{\sum_{j=1}^{k}\mathbb{P}_s(\mathbf{x}_j)}
        =\sum_{\mathbf{g}\in{F^{-1}(\mathcal{X})}}\mathbb{Q}_c(\mathbf{g})\log\frac{\mathbb{Q}_c(\mathbf{g})}{\mathbb{Q}_s(\mathbf{g})}=D_{KL}(\mathbb{Q}_c||\mathbb{Q}_s).
\label{proof:kl-f}
\end{equation}
As the denoising function $F$ is deterministic and invertible, we can use the same method described above to obtain 
\begin{equation}
    D_{KL}(\mathbb{Q}_c||\mathbb{Q}_s)\geq D_{KL}(\mathbb{P}_c||\mathbb{P}_s).
    \label{proof:kl-b}
\end{equation}
  Combining \Cref{proof:kl-f} and \Cref{proof:kl-b}, we have $D_{KL}(\mathbb{P}_c\|\mathbb{P}_s)=D_{KL}(\mathbb{Q}_c\|\mathbb{Q}_s)$ hold. The proof is complete.
\end{proof}
\Cref{theorem:t0} shows that, the KL divergence between the distributions of normal and stego images is theoretically preserved in the noise space of diffusion models. This means that any steganographic operation disrupting the original noise distribution of diffusion models will inevitably introduce a distributional divergence between stego and normally generated images. According to the steganographic security definition in \Cref{eq:security1}, such an operation increases $D_{KL}(\mathbb{P}_c \parallel \mathbb{P}_s)$, thereby reducing security. In other words, DM-GIS can be theoretically undetectable only if their message-mapping schemes do not alter the diffusion model's natural noise distribution.
\begin{figure}
  \centering
\includegraphics[width=.5\textwidth]{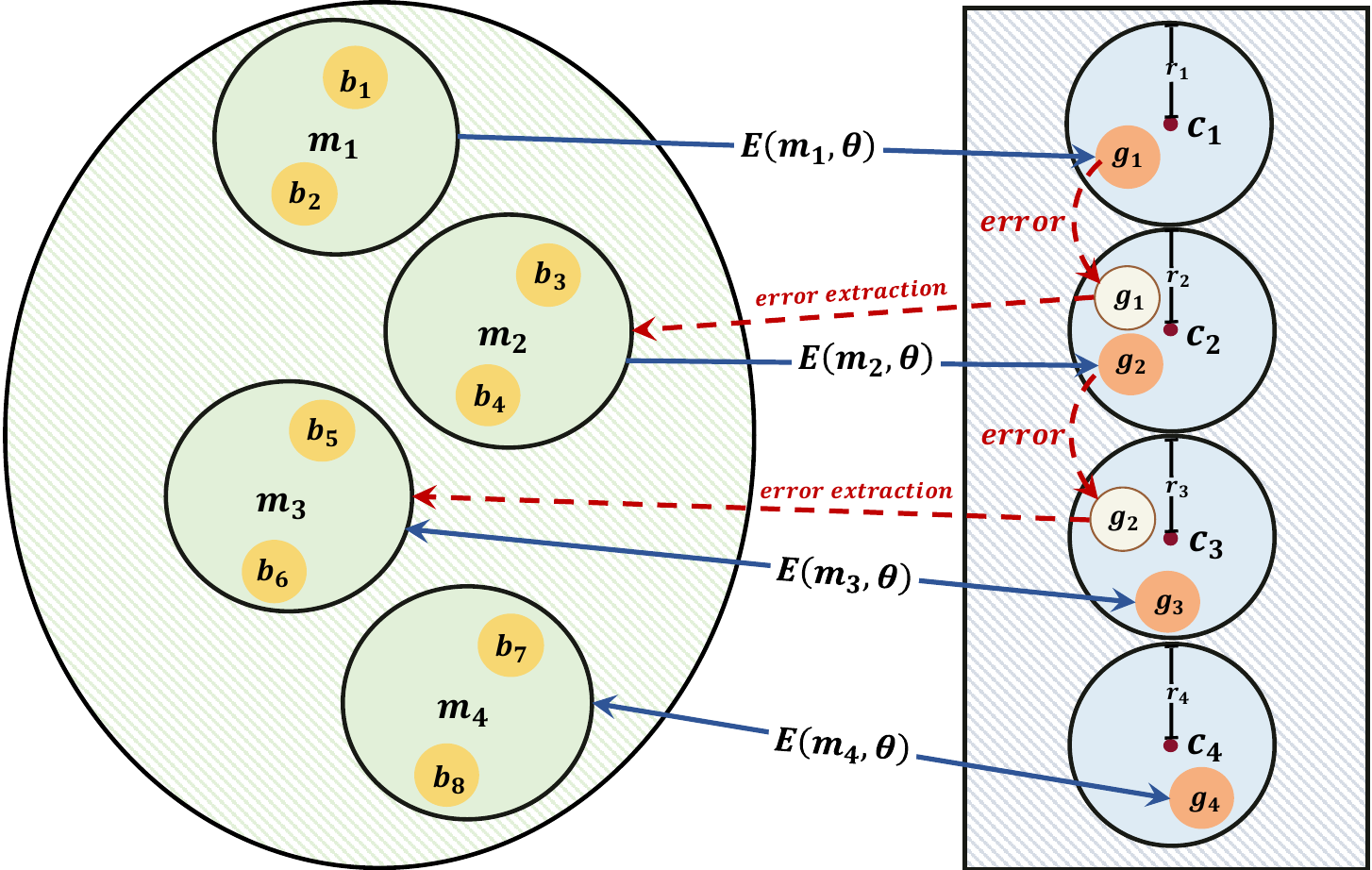}
   \caption{An illustrative example of $E(\mathbf{m}, \theta) = \mathbf{g}$ with $l = 2$, showcasing an invertible bit-to-noise mapping for DM-GIS while highlighting a scenario where extraction errors arise.}
   \label{fig:proof}
\end{figure}
\subsection{Connecting Message Extraction Accuracy to DM-GIS Security}
\label{subsubsec:security-ea}
Beyond steganographic security, message extraction accuracy is another critical metric considered by steganographers. We explore the relationship between DM-GIS security and message extraction accuracy as below. 

Given an encrypted secret bit sequence $\mathbf{b}=b_1b_2\cdots b_n$, the sequence $\mathbf{b}$ can be further partitioned into multiple subsequences of length $l$ ($l\geq1$)\cite{gaussian-shading}. These subsequences collectively form an integer sequence $\mathbf{m}=m_1m_2\cdots m_k$, where $m_i\in \{0,1,\cdots,2^l-1\}$. 
\begin{proposition}\label{prop:one}
Let $E(\mathbf{m}; \theta) = \mathbf{g}$ denote the steganographic encoder employed in DM-GIS, where $\mathbf{g} = (g_1, g_2, \dots, g_k) \sim \mathbb{Q}_s$ and $\theta$ denotes an adjustable parameter. The security of DM-GIS is quantified by $D_{KL}(\mathbb{Q}_c \,\|\, \mathbb{Q}_s) = \epsilon$. If $\theta$ is modified to $\theta'$ in order to enhance message extraction accuracy, resulting in $E(\mathbf{m}; \theta') = \mathbf{g}'$, then $D_{KL}(\mathbb{Q}_c \,\|\, \mathbb{Q}_s) > \epsilon$.
\end{proposition}
\begin{proof}
    The proof can be formulated within the framework of coding theory. Specifically, the steganographic encoder can be expressed as $E(m_i, \theta) = g_i = c_j + \eta_j$, where $j \in \{1, \dots, 2^l - 1\}$, which defines a Hamming ball centered at $c_j$ with radius $r$. To extract secret bits, $E^{-1}$ must be at least surjective, but not necessarily injective. \Cref{fig:proof} provides an illustrative example: all elements $g_i$ are projected into the blue disjoint regions, ensuring that any $g_i$ satisfying $\|\eta_j\| \leq r$ can be decoded into a unique bit sequence. However, when perturbations occur, such as when $g_1$ falls within the region centered at $c_2$, a message extraction error occurs. To reduce the error, two strategies can be employed: increasing the distance between the center points $c_j$, or reducing $r$ of each blue region. Formally speaking, an alternative encoder parameter $\theta'$ can be chosen such that $E(\mathbf{m}, \theta') = \mathbf{g}'$ improves message extraction. However, according to \Cref{theorem:t0}, this inevitably alters the distribution of $\mathbf{g}$, leading to $D_{KL}(\mathbb{Q}_c \| \mathbb{Q}_s) > \epsilon$. This completes the proof.
\end{proof}

According to \Cref{theorem:t0}, \Cref{prop:one} implies that, for a given steganographic encoder $E(\cdot;\theta)$, any adjustment of $\theta$ that improves message extraction accuracy will inevitably compromise the security of DM-GIS. In other words, once $E(\cdot;\theta)$ is fixed, DM-GIS cannot simultaneously achieve both high extraction accuracy and strong security. This observation makes \Cref{prop:one} particularly useful for analyzing the security of adjustable DM-GIS methods \cite{diffstega,ldstega}. To overcome this inherent trade-off, a more refined design of the encoder $E$ is therefore required.
\begin{figure}
  \centering
\includegraphics[width=\textwidth]{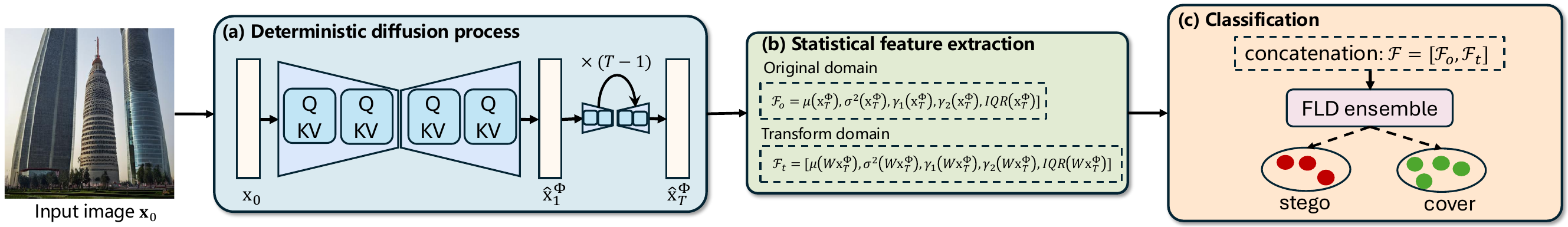}
   \caption{\textbf{Detection pipeline of the proposed NS-DSer.} The input image $\mathbf{x}_0$ first undergoes a deterministic condition-free diffusion process which is implemented by the ODE solver $\Phi$, to obtain the noise $\hat{\mathbf{x}}_{T}^{\Phi}$. Five statistical features of $\hat{\mathbf{x}}_{T}^{\Phi}$—mean, variance, skewness, kurtosis, and IQR—are extracted to form feature $\mathcal{F}_o$. Additionally, a transform-domain feature $\mathcal{F}_t$ is extracted from $\hat{\mathbf{x}}_{T}^{\Phi}$. The features $\mathcal{F}_0$ and $\mathcal{F}_t$ are concatenated and fed into a binary classifier.}
   \label{fig:ns_dser_framework}
\end{figure}
\section{Noise Space-based Diffusion Steganalyzer}\label{sec:nsdser}
\Cref{theorem:t0} demonstrates that the KL divergence between cover and stego image distributions is preserved in the noise space of diffusion models. This implies steganalysis can theoretically be performed in the noise space. This procedure is far more tractable than steganalysis in the image space because it can be reduced to the problem of distinguishing stego noise from Gaussian noise. However, in practice, $F$ is only approximately invertible due to numerical errors in \Cref{eq:pf-ode-de} and image quantization. Furthermore, \Cref{theorem:t0} holds only under the assumption that $F$ is known, which is impractical in real-world detection scenarios. That said, recent work \cite{gaussian-shading} has empirically demonstrated that using samplers and condition guidance different from those in stego image generation still achieves high message extraction accuracy, suggesting the feasibility of noise space-based steganalysis for DM-GIS methods. Motivated by this, we propose NS-DSer, which comprises three key steps.

\noindent\textbf{Deterministic Diffusion Process}. We select the noise space corresponding to the start timestep \( T \) of the denoising process for steganalysis. Specifically, given an image \( \mathbf{x}_0 \), we first iteratively diffuse it to the noise \( \hat{\mathbf{x}}_T^{\Phi} \) using an ODE solver \( \Phi \) with \( N_s \) sampling steps:
\begin{equation}
    \label{eq:ode_solver}
    \hat{\mathbf{x}}_{t+1}^{\Phi} = \mathbf{x}_t + \int_{t}^{t+1} \left( f(t)\mathbf{x}_t + \frac{g^2(t)}{2\sigma_t} \boldsymbol{\epsilon}_{\phi}(\mathbf{x}_t,t,\mathbf{c}) \right) \mathrm{d}t 
    \approx \Phi(\mathbf{x}_t, t, t+1, \mathbf{c};\phi),
\end{equation}
where $\mathbf{c}$ is a given condition (\emph{e.g.}, text). In this paper, to balance computational efficiency and numerical accuracy, we adopt $2^{\text{nd}}$ order ODE solvers, such as DPM-Solver-2 \cite{dpm} and Heun solver \cite{karras}, to implement $\Phi(\cdot,\cdot,\cdot,\cdot;\phi)$. 

When Classifier-free Guidance (CFG) \cite{cfg} is applied, $\boldsymbol{\epsilon}_{\phi}(\mathbf{x}_t, t,\mathbf{c})$ in \Cref{eq:ode_solver} is converted into 
\begin{equation}
\tilde{\boldsymbol{\epsilon}}_{\phi}(\mathbf{x}_t, t,\mathbf{c})=\boldsymbol{\epsilon}_{\phi}(\mathbf{x}_t, t,\emptyset) + \omega(\boldsymbol{\epsilon}_{\phi}(\mathbf{x}_t, t,\mathbf{c})+\boldsymbol{\epsilon}_{\phi}(\mathbf{x}_t, t,\emptyset)),   
\end{equation}
where $\omega$ denotes the CFG scale. In practical detection scenarios, the specific conditions $\mathbf{c}$ used by steganographers are often unknown. Therefore, we can leverage off-the-shelf vision language models (VLM), such as ChatGPT \cite{gpt-4o}, to extract $\mathbf{c}$ from images by simply prompting VLMs with ``Please briefly describe this image.''. Alternatively, we may choose not to apply $\mathbf{c}$ in the deterministic diffusion process. A comparison of the two approaches is presented in \Cref{subsec:ablation_study}.

\noindent\textbf{Statistical Feature Extraction}. In the target noise space, $\hat{\mathbf{x}}_T^{\Phi}$ approximates random noise, meaning it contains little to no semantic information from the original image $\mathbf{x}_0$. Therefore, the deep network architectures commonly used in previous image steganalysis methods are not well-suited for this setting. Additionally, we observe that directly applying these networks in the noise space can easily lead to overfitting, resulting in unstable training. To address these challenges, we propose using statistical features of $\hat{\mathbf{x}}_T^{\Phi}$ for steganalysis. Specifically, we find that $\hat{\mathbf{x}}_T^{\Phi}$ approximates standard Gaussian noise when the selected diffusion model $\boldsymbol{\phi}$ is variance-preserving \cite{scored-generative,karras}. Therefore, distinguishing between cover and stego images can be viewed as differentiating the stego noise distribution from the known cover noise distribution. Based on this observation, we first flatten $\hat{\mathbf{x}}_T^{\Phi}$ and extract five statistical features—mean, variance, skewness, kurtosis, and interquartile range (IQR)—to represent $\hat{\mathbf{x}}_T^{\Phi}$ as follows:
\begin{equation}
    \label{eq:ori_fea_ext}
    \mathcal{F}_o = [\mu(\hat{\mathbf{x}}_T^{\Phi}),\sigma^2(\hat{\mathbf{x}}_T^{\Phi}),\gamma_1(\hat{\mathbf{x}}_T^{\Phi}),\gamma_2(\hat{\mathbf{x}}_T^{\Phi}),IQR(\hat{\mathbf{x}}_T^{\Phi})]
\end{equation}

Similar to traditional image steganography, where secret messages are embedded in the frequency domain, some DM-GIS methods \cite{improved,establishing} hide secret bits in the transform domain of $\mathbf{x}_T$. In this paper, we select the Discrete Cosine Transform and apply it to each channel of $\hat{\mathbf{x}}_T^\Phi$, denoted simply as $W\hat{\mathbf{x}}_T^\Phi$. To capture the differences in statistical distributions between cover and stego examples in the transform domain, we similarly extract five statistical features with
\begin{equation}
    \label{eq:td_fea_ext}
    \mathcal{F}_t = [\mu(W\hat{\mathbf{x}}_{T}^{\Phi}),\sigma^2(W\hat{\mathbf{x}}_{T}^{\Phi}),\gamma_1(W\hat{\mathbf{x}}_{T}^{\Phi}),\gamma_2(W\hat{\mathbf{x}}_{T}^{\Phi}), 
    IQR(W\hat{\mathbf{x}}_{T}^{\Phi})].
\end{equation}
The final feature is represented by $\mathcal{F} = [\mathcal{F}_o, \mathcal{F}_t]$.

\noindent\textbf{Classification} For classification, we use the Fisher Linear Discriminant (FLD) ensemble \cite{kodovsky2011ensemble} due to its excellent balance between performance and computational efficiency. The training procedure for NS-DSer is provided in \Cref{alg:ns-dser} and \Cref{fig:ns_dser_framework}.
\begin{algorithm}[!t]
  \caption{Training Procedure of NS-DSer}
  \raggedright
\textbf{Input}: cover images $\mathbf{x}_c$, stego images $\mathbf{x}_s$, sampling steps $N_s$, guidance scale $\omega$, ODE solver $\Phi(\cdot,\cdot,\cdot,\cdot;\phi,\omega)$. \\
\textbf{Output}: The optimal FLD ensemble.
  \begin{algorithmic}[1]
        \FOR{$i = 1$ to $N_s$}
             \STATE $\hat{\mathbf{x}}_{c,t+1}^{\Phi} = \Phi(\hat{\mathbf{x}}_{c,t}^{\Phi}, t, t+1, \varnothing;\phi,\omega)$
             \STATE $\hat{\mathbf{x}}_{s,t+1}^{\Phi} = \Phi(\hat{\mathbf{x}}_{s,t}^{\Phi}, t, t+1, \varnothing;\phi,\omega)$
        \ENDFOR
        \STATE Flatten $\hat{\mathbf{x}}_{c,T}^{\Phi}$ and $\hat{\mathbf{x}}_{s,T}^{\Phi}$ 
        \STATE Extract feature $\mathcal{F}_o^c$ of $\hat{\mathbf{x}}_{c,T}^{\Phi}$ with \Cref{eq:ori_fea_ext}
        \STATE Extract feature $\mathcal{F}_o^s$ of $\hat{\mathbf{x}}_{s,T}^{\Phi}$ with \Cref{eq:ori_fea_ext}
        \STATE Extract feature $\mathcal{F}_t^c$ of $W_q\hat{\mathbf{x}}_{c,T}^{\Phi}$ with \Cref{eq:td_fea_ext}
        \STATE Extract feature $\mathcal{F}_t^s$ of $W_q\hat{\mathbf{x}}_{s,T}^{\Phi}$ with \Cref{eq:td_fea_ext}
        \STATE $\mathcal{F}_c=[\mathcal{F}_o^c, \mathcal{F}_t^c]$ and $\mathcal{F}_s=[\mathcal{F}_o^s, \mathcal{F}_t^s]$
        \STATE Train FLD ensemble on features $\mathcal{F}_c$ and $\mathcal{F}_s$ 
\STATE Return the optimal FLD ensemble
  \end{algorithmic}
  \label{alg:ns-dser}
  \end{algorithm}

\section{Experiments}\label{sec:exp_results}
  \begin{table}[!t]
      \centering
      \caption{Setups for \emph{Scenarios \#1\(\sim\)\#4}. $r_g$ and $r_n$ represent the ratios of generated and natural images in the cover (or stego) image dataset, respectively. SD1.5, SD2.1, and DS7 are abbreviations of Stable Diffusion 1.5, Stable Diffusion 2.1, and Dreamshaper 7, respectively. $1^{\text{st}}$ order sampler adopts 50-step sampling, with 20 steps for the $2^{\text{nd}}$ order. $\dagger$ denotes the setups for CRoSS.}
      \label{tab:setups}
      \begin{tabular}{ccccccc}
      \toprule
       & & $r_g$ & $r_n$ & Backbone  & Sampling steps & $\omega$  \\
      \midrule
      \multirow{2}{*}{\emph{\#1}}
      &\emph{cover} & 100\% & 0\% & SD1.5 & 20, 50 & 7.5, 1$^{\dagger}$ \\
      &\emph{stego} & 100\% &0\% & SD1.5 &  20, 50 & 7.5, 1$^{\dagger}$ \\
      \midrule
      \multirow{2}{*}{\emph{\#2}}
      &\emph{cover} & 50\% & 50\% & SD1.5 & 20, 50 &7.5, 1$^{\dagger}$  \\
      &\emph{stego} & 100\% & 0\% & SD1.5 & 20, 50 &7.5, 1$^{\dagger}$ \\
      \midrule
      \emph{\#3}&\emph{cover} & 50\% & 50\% & SD1.5,SD2.1,DS7 & $[15,\!25]$,$[45,\!55]$ &$[4.5,\!10.5]$,$[1\!,\!9]^\dagger$  \\
      \emph{\#4}&\emph{stego} & 100\% & 0\% & SD1.5,SD2.1,DS7 &  $[15,\!25]$,$[45,\!55]$ &$[4.5,\!10.5]$,$[1\!,\!9]^\dagger$\\
      \bottomrule
      \end{tabular}
  \end{table}

   We reevaluate the security of nine existing DM-GIS methods—MN \cite{diffusion-stego}, MC \cite{diffusion-stego}, MB \cite{diffusion-stego}, Gaussian-Shading (G-S) \cite{gaussian-shading}, StegaDDPM \cite{stegaddpm}, LDStega \cite{ldstega}, GSD \cite{improved}, mas-GRDH \cite{establishing}, and CRoSS \cite{cross}—via five steganalyzers: XuNet \cite{xunet}, SRNet \cite{srnet}, SuaStegBet \cite{you2020siamese}, UCNet \cite{ucnet}, and NS-DSer.
  
  \noindent \textbf{DM-GIS Setups}. For MN, MC, MB, G-S, GSD, and mas-GRDH, we use the 2$^\text{nd}$ order ODE solver \cite{karras} for stego image generation, and the 1$^\text{st}$ sampler \cite{ddim} for the others. The truncated interval of LDStega is set to 0.3. For other methods, the default settings are adopted unless otherwise specified.
  
  \noindent \textbf{Steganalyzer Setups}. The first four steganalyzers are each trained on 6,000 cover-stego image pairs and tested on 1,000 pairs over 100 epochs. In contrast, owing to the limited 10-dimensional feature space extracted by NS-DSer, it is trained on a smaller dataset of 1,800 cover-stego pairs and tested on 200 pairs. All training and testing images are generated based on prompts from the Flickr8K dataset \footnote{https://www.kaggle.com/datasets/adityajn105/flickr8k}. For NS-DSer, given that virtually all existing DM-GIS methods leverage U-Net-based diffusion models for steganography, we also adopt this architecture for feature extraction. Unless otherwise specified, we adopt Stable Diffusion 2.1 as the base generative model and employ 20-step Heun sampling \cite{karras} without prompt guidance. NS-DSer is implemented in both PyTorch and MATLAB, with all experiments conducted using 8 A100 GPUs and 8 L40 GPUs.
  \subsection{Security and Detection Evaluation}
  To comprehensively evaluate the security of DM-GIS methods and compare the performance of current steganalyzers, we set up the following four detection scenarios with increasing difficulty levels.
  
  \noindent \textbf{\emph{Scenario \#1}}. In this scenario, steganographers utilize the same DM-GIS method for covert communication. Cover images are generated from the standard normal distribution, employing the identical set of sampling configurations as those used in stego image generation. Specific parameters are detailed in \Cref{tab:setups}. The passive adversary \(\mathcal{A}\) performs stego image detection using each steganalyzer.
  
  \noindent \textbf{\emph{Scenario \#2}}.
  Here, steganographers also adopt the same DM-GIS method for covert communication. However, the cover image set is expanded to include not only generated images but also natural images sourced from the Flickr8K dataset, with specific details provided in \Cref{tab:setups}.
  
  \noindent \textbf{\emph{Scenario \#3}}.
  In this scenario, steganographers continue to use the same DM-GIS method for covert communication, with cover images encompassing both generated and natural images. The key distinction lies in the fact that stego images and generated cover images are generated using different generative backbones, sampling steps, and CFG scales \(\omega\), as outlined in \Cref{tab:setups}.
  
  \noindent \textbf{\emph{Scenario \#4}}.
  In this scenario, steganographers are permitted to employ different DM-GIS methods for covert communication. Cover images include both generated and natural images, while stego images and generated cover images differ in generative backbones, sampling steps, and CFG scales \(\omega\), as shown in \Cref{tab:setups}. Specifically, we assume steganographers utilize four DM-GIS methods for stego image generation: MN, MC, MB, and GSD.
  
  \begin{table*}[!t]
    \centering
    \caption{Detection accuracy (\%) and overal evaluation of each steganalyzer on nine DM-GIS methods across \emph{Scenarios \#1$\sim$\#3}. }
    \label{tab:s1-3}
    \setlength{\tabcolsep}{4pt}
    \begin{tabular}{clccccccccc|c}
    \toprule
    &\emph{Scenario \#1} & MN & MC  & MB & G-S & StegaDDPM  & LDStega & GSD  & mas-GRDH & CRoSS & Overall \\
    \midrule
    \multirow{5}{*}{\emph{\#1}}
    &XuNet          & 49.17 & 52.25 & 58.86  & 51.38 & 50.67 & 89.41 &50.79  & 50.88 & 83.63 &0.8255\\
    &SRNet          & 49.71 & 80.95 & 97.29 & 50.21 & 51.92 & 98.75 &49.25  & 51.88 & 88.40 &0.9039\\
    &SiaStegNet     & 50.33 & 91.52 & 98.71 & 50.75 & 50.96 & 99.90 &50.38  & 50.13 & 93.20 &0.9283\\
    &UCNet          & 49.67 & 92.90 & 99.60 & 51.50 & 50.54 & 99.95 & 49.67 & 50.70 & \textbf{93.65} &0.9281\\
    \rowcolor{gray!30}
    &NS-DSer (ours) & 52.21  & \textbf{99.78}  & \textbf{99.99}  & 50.95 & 50.05 & \textbf{99.99} &\textbf{99.58}  & 51.13  & 85.70&\textbf{0.9742}\\
    \bottomrule
    \multirow{5}{*}{\emph{\#2}}
    &XuNet               & 59.68 & 61.65 & 70.21 & 60.18  & 60.35 &  91.25 & 60.21 & 60.19 & 81.13 &0.7635\\
    &SRNet               & 70.28 & 86.95 & 97.90 & 70.80  & 70.25 &  99.96 & 70.25 & 70.00 & 87.50 &0.7392\\
    &SiaStegNet          & 74.10 & 91.80 & 98.55 & 74.20  & 73.95 &  99.95 & 73.00 & 74.00 & 91.45 &0.7141\\
    &UCNet               & 73.10 & 90.55 & 99.15 & 74.05  & 73.50 &  99.90 & 73.80 & 74.10 & \textbf{93.50} &0.7200\\
    \rowcolor{gray!30}
    &NS-DSer (ours)      & 62.42 &  \textbf{99.73}&  \textbf{99.99} & 62.27& 62.83& \textbf{99.99}   & \textbf{99.60} &  62.98   & 90.28 & \textbf{0.8633}\\
    \bottomrule
    \multirow{5}{*}{\emph{\#3}}&XuNet              
                         & 61.14 & 62.94 & 61.96 & 58.50 & 60.19 & 89.73 & 59.39  & 59.68 &  72.40 &0.7476  \\
    &SRNet               & 69.05 & 76.60 & 92.80 & 69.06 & 69.30 & 99.55 & 69.00  & 69.20 &  77.90 & 0.7243 \\
    &SiaStegNet          & 73.75 & 83.40 & 94.40 & 73.10 & 73.95 & 99.90 & 73.50  & 73.00 &  77.40 & 0.6941 \\
    &UCNet               & 73.85 & 86.35 & 96.20 & 74.35 & 74.15 & 99.80 & 74.65  & 74.45 &  79.85 & 0.6949 \\
    \rowcolor{gray!30}
    &NS-DSer (ours)      & 59.50 &\textbf{98.18} & \textbf{99.05} & 59.45 &59.57 & \textbf{99.99}  & \textbf{96.15}& 59.53  & \textbf{83.48}& \textbf{0.8817}\\
    \bottomrule
    \end{tabular}
    \end{table*}
    \begin{table}[!t]\footnotesize
        \centering
        \caption{Detection accuracy (\%) of each steganalyzer under \emph{Scenario \#4}.}
        \label{tab:s4}
        \setlength{\tabcolsep}{5pt}
        \begin{tabular}{lccccc}
        \toprule
        \emph{\#4} & XuNet  & SRNet  & SiaStegNet & UCNet  & NS-DSer (ours) \\
        \midrule
        & 57.96 & 74.80 & 75.80 & 73.00 & \textbf{86.26} \\
        \bottomrule
        \end{tabular}
        \end{table}
  \noindent \textbf{Results of \emph{Scenario \#1}}. In evaluating DM-GIS methods, to our knowledge, MN, G-S, StegaDDPM, and mas-GRDH all fall into the category of distribution-preserving methods. Theoretically, they have been proven to be statistically indistinguishable from generated cover images. Put simply, detecting these methods boils down to distinguishing between encrypted 0-1 sequences and pseudorandom 0-1 sequences. Clearly, such a problem cannot be solved in polynomial time; hence, in accordance with \Cref{eq:security2}, they are computationally secure. The experimental results in \Cref{tab:s1-3} also align with this conclusion. As for other DM-GIS methods, they can be detected by at least one steganalyzer, suggesting that their security under \emph{Scenario \#1} is not as claimed.
  
  Regarding steganalyzers, XuNet fails to detect all DM-GIS methods except LDStega and CRoSS. DM-GIS methods that SRNet, SiaStegNet, and UCNet can detect is also limited. In contrast, NS-DSer effectively detects nearly all DM-GIS methods except for distribution preserving ones. This underscores a key limitation of previous steganalyzers: their reliance on the image space, which is insufficient in detecting DM-GIS methods. NS-DSer addresses it by conducting steganalysis in the noise space. However, we observe that NS-DSer performs relatively poorly in detecting CRoSS under \emph{Scenario \#1}. This is because it focuses more on distribution differences rather than slight perturbations. That said, such a design enables faster training, as demonstrated in \Cref{tab:time}.
  
  \noindent \textbf{Results of \emph{Scenario \#2}}. \emph{Scenario \#2} includes natural images as cover images for both training and testing. The evaluation results for detecting MC, MB, LDStega, GSD, and CRoSS are comparable to those in \emph{Scenario \#1}. However, when detecting distribution-preserving methods, all steganalyzers exhibit a significant improvement in detection accuracy. As elaborated in \emph{Scenario \#1}, distribution-preserving methods are statistically indistinguishable from generated cover images. \textbf{The heightened detection accuracy for these methods inherently suggests that the five steganalyzers have not learned to differentiate between cover and stego images, but rather to distinguish between natural images and generated ones}. Ideally, a well-performing steganalyzer should yield poor results (with detection accuracy close to 50\%) when dealing with distribution-preserving methods, while achieving high accuracy (approaching 100\%) for other methods. Guided by this principle, we present a comprehensive evaluation of the five steganalyzers across nine DM-GIS methods as follows: For detection results on distribution-preserving methods, denoted as \(D_{dp}\), we first normalize them using the formula \(1-\frac{|D_{dp}-50|}{50}\). For results on other methods, denoted as \(D_{ndp}\), the normalization is performed as \(1-\frac{|D_{ndp}-100|}{100}\). Subsequently, we compute the average of all normalized results. As shown in \Cref{tab:s1-3}, NS-DSer outperforms the other four steganalyzers in terms of overall detection performance. 
  
  \noindent \textbf{Results of \emph{Scenario \#3}}. As \Cref{tab:s1-3} shows, conventional steganalyzers perform notably worse, revealing their struggle with heterogeneous data. In contrast, our NS-DSer degrades minimally, maintaining over 95\% accuracy for MC, MB, LDStega, and GSD—verifying its robustness to such data.
  
  \noindent \textbf{Results of \emph{Scenario \#4}}. \emph{Scenario \#4} simulates a steganographic scenario closer to real-world conditions. As \Cref{tab:s4} shows, conventional image steganalyzers again perform significantly worse, while NS-DSer outperforms them dramatically—further validating its resilience to heterogeneous data. These results underscore NS-DSer’s robustness and practicality in realistic scenarios.
  
  \begin{table}[!t]
  \centering
  \caption{Ablation study of sampling step number $N_s$ across \emph{Scenarios \#1 $\sim$ \#3.}}
  \label{tab:ab1}
  \begin{tabular}{ccccccc}
  \toprule
   & & MC  & MB  & LDStega & GSD  & CRoSS \\
  \midrule
  \multirow{3}{*}{\emph{\#1}}
  &$N_s=15$  & 99.78  & 99.99  & 99.99 & 99.43 & 85.72\\
  &$N_s=20$ & 99.78 & 99.99 & 99.99 & 99.58 & 85.70 \\
  &$N_s=25$ &99.82 & 99.99 & 99.99 & 99.42 & 85.52 \\
  \midrule
  \multirow{3}{*}{\emph{\#2}}&
  $N_s=15$  & 99.67 & 99.99 & 99.99 & 99.52 & 89.30 \\
  &$N_s=20$ & 99.73 &99.99  & 99.99 & 99.60 & 90.28 \\
  &$N_s=25$ & 99.73& 99.99 & 99.99 & 99.77 & 90.28 \\
  \midrule
  \multirow{3}{*}{\emph{\#3}}
  &$N_s=15$ & 98.43 & 99.12 & 99.99 & 96.65 & 83.08\\
  &$N_s=20$ & 98.18 & 99.05  & 99.99 & 96.15 & 83.48 \\
  &$N_s=25$ & 98.20 & 99.20 &99.99  & 97.53 & 83.22 \\
  \bottomrule
  \end{tabular}
  \end{table}
  \begin{table}[!t]
  \centering
  \caption{Ablation study on detection with/without prompt guidance across \emph{Scenarios \#1 $\sim$ \#3}.}
  \label{tab:ab2}
  \begin{tabular}{ccccccc}
  \toprule
   & & MC  & MB  & LDStega & GSD  & CRoSS \\
  \midrule
  \multirow{2}{*}{\emph{\#1}}
  &\emph{w/o cond.} & 99.78 & 99.99 & 99.99 & 99.58 & 85.70  \\
  &\emph{w/ cond.} & 99.60 & 99.99 & 99.99 & 99.50 & 75.18 \\
  \midrule
  \multirow{2}{*}{\emph{\#2}}
  &\emph{w/o cond.} & 99.73 &99.99  & 99.99 & 99.60 & 90.28  \\
  &\emph{w/ cond.} & 99.67 & 99.88 & 99.99 & 99.48 & 78.00 \\
  \midrule
  \multirow{2}{*}{\emph{\#3}}
  &\emph{w/o cond.} & 98.18 & 99.05  & 99.99 & 96.15 & 83.48  \\
  &\emph{w/ cond.} & 99.35 & 99.92 & 99.99 & 99.60 & 77.12 \\
  \bottomrule
  \end{tabular}
  \end{table}
  \begin{table}[!t]
  \centering
  \caption{Ablation study of sampling steps $N_s$ and prompt guidance under \emph{Scenario \#4}.}
  \label{tab:ab3}
  \begin{tabular}{cccc|cc}
  \toprule
   \emph{\#4} & $N_s=15$  & $N_s=20$  & $N_s=25$ & \emph{w/o cond.}  & \emph{w/ cond.} \\
  \midrule
      & 86.25 & 86.26 & 86.00& 86.26 & 86.90\\
  \bottomrule
  \end{tabular}
  \end{table}
  \begin{table}[!t]
      \centering
      \caption{Training duration of each steganalyzer.}
      \label{tab:time}
      \begin{tabular}{l|ccccc}
      \toprule
      Time (h) & XuNet  & SRNet  & SiaStegNet & UCNet  & NS-DSer (ours) \\
      \midrule
      & 1.84  & 5.52 & 7.90 & 8.25 & 1.11 \\
      \bottomrule
      \end{tabular}
      \end{table}
  \subsection{Ablation Study}
  \label{subsec:ablation_study}
  \noindent \textbf{Sampling Steps $N_s$}. We use three \(N_s\) values in the deterministic diffusion sampling of NS-DSer to evaluate their impact on detection. Results across \emph{Scenarios \#1\(\sim\)\#4}, as reported in \Cref{tab:ab1} and \Cref{tab:ab3}, consistently show that NS-DSer's detection performance is insensitive to  \(N_s\).
  
  \noindent \textbf{Prompt Guidance}. We use the DouBao 1.5 Vision to extract image descriptions as prompt guidance for deterministic diffusion sampling. As \Cref{tab:ab2} shows, prompt guidance barely impacts NS-DSer's detection of MC, MB, GSD, and LDStega but significantly affects CRoSS detection. However, in practice, the guidance prompts used by steganographers are usually inaccessible, so we are not responsible for this performance drop.
  \subsection{Computational Efficiency}
  \Cref{tab:time} presents the time consumption of training each steganalyzer on a single L40 GPU. Specifically, XuNet, SRNet, SiaStegNet, and UCNet are trained for 100 epochs. As shown, training NS-DSer takes the least training time, which is attributed to our novel steganalysis framework. This framework avoids performing steganalysis in complex high-dimensional image spaces and instead shifts steganalysis to low-dimensional noise spaces, thus achieving highly efficient training. Moreover, NS-DSer adopts condition-free deterministic sampling for feature extraction (as shown in \Cref{fig:framework}), further accelerating the training process.

\section{Conclusion}
\label{sec:conclusion}
In this paper, we rethink the security of DM-GIS and analyze, from both theoretical and experimental perspectives, how altering the noise distribution in diffusion models impacts DM-GIS security. We introduce NS-DSer, the first dedicated steganalysis framework for detecting DM-GIS. By shifting steganalysis from the complex high-dimensional image space to the lower-dimensional noise space of diffusion models, NS-DSer enables efficient training and testing. We evaluate NS-DSer across four increasingly challenging detection scenarios, and experimental results show it outperforms traditional image steganalyzers across all scenarios overall, bringing it closer to practical steganalysis. In the future, we will refine our feature extraction procedure to further enhance detection performance.

\bibliographystyle{unsrt}  
\bibliography{references}

\end{document}